\documentclass{scrartcl}

\usepackage[T1]{fontenc}
\usepackage[utf8]{inputenc}

\usepackage{amsmath}
\usepackage{amssymb}
\usepackage{bm}
\usepackage{mathrsfs}
\usepackage{csquotes}
\usepackage{enumitem}

\usepackage{tikz}
\usetikzlibrary{positioning}
\usetikzlibrary{automata}
\usetikzlibrary{trees}
\usetikzlibrary{calc}
\usetikzlibrary{decorations.pathmorphing}

\usepackage{subcaption}

\usepackage{amsthm}
\usepackage{thmtools}

\theoremstyle{plain}
\newtheorem{theorem}{Theorem}[section]

\newtheorem{proposition}[theorem]{Proposition}

\newtheorem{fact}[theorem]{Fact}

\theoremstyle{definition}
\newtheorem{definition}[theorem]{Definition}
\newtheorem{example}[theorem]{Example}

\theoremstyle{remark}
\newtheorem{remark}[theorem]{Remark}
\newtheorem*{remark*}{Remark}

\usepackage{calc}
\newlength{\edgelength}

\newcommand{\aedge}[3]{%
  \begin{tikzpicture}[auto, shorten >=1pt, >=latex, baseline=(l.base), inner sep=0pt, outer xsep=0.3333em]
    \node (l) {\ensuremath{#1}};%
    \setlength{\edgelength}{\widthof{\scriptsize\ensuremath{#2}}+0.5cm}%
    \node[base right=\edgelength of l] (r) {\ensuremath{#3}};%
    \path[->] (l.mid east) edge node[inner sep=0pt, yshift=0.2em] {\scriptsize\ensuremath{#2}} (r.mid west);%
  \end{tikzpicture}%
}

\newcommand{\laedge}[3]{%
  \begin{tikzpicture}[auto, shorten >=1pt, >=latex, baseline=(l.base), inner sep=0pt, outer xsep=0.3333em]
    \node (l) {\ensuremath{#1}};%
    \setlength{\edgelength}{\widthof{\scriptsize\ensuremath{#2}}+0.5cm}%
    \node[base right=\edgelength of l] (r) {\ensuremath{#3}};%
    \path[<-] (l.mid east) edge node[inner sep=0pt, yshift=0.2em] {\scriptsize\ensuremath{#2}} (r.mid west);%
  \end{tikzpicture}%
}

\tikzset{path/.style = {decorate, decoration={snake, pre length=1mm, post length=2mm, amplitude=0.3mm, segment length=1mm}}}

\usepackage{tabularx}
\newcommand{\problem}[3][]{%
  \par\vspace{0.125cm plus 0.1cm minus 0.05cm}\begin{tabularx}{\textwidth-2\parindent}{lX}%
    \if\relax\detokenize{#1}\relax%
    \else%
      \textnormal{\textbf{Constant:}}&#1\\%
    \fi%
    \textnormal{\textbf{Input:}}&#2\\%
    \textnormal{\textbf{Question:}}&#3\\%
  \end{tabularx}\vspace{0.125cm plus 0.1cm minus 0.05cm}\par%
  }
  
\newcommand{\changefont}[3]{\fontfamily{#1}\fontseries{#2}\fontshape{#3}\selectfont}
\newcommand*{\DecProblem}[1]{{\changefont{cmtt}{m}{sc}#1}}
\usepackage{algorithm}
\usepackage{listings}

\lstnewenvironment{pseudocode}[1][]
{
  \lstset{
    mathescape=true,
    escapeinside={(*@}{@*)},
    numbers=left,
    numberstyle=\tiny,
    basicstyle=\scriptsize\ttfamily,
    keywordstyle=\fontencoding{T1}\fontfamily{lmtt}\bfseries,
    keywords={var, begin, end, return, procedure, fun, if, then, else, fi, for, each, while, do, od, true, false, and, or, not, break, accept, reject, assert, guess}
    numbers=left,
    xleftmargin=.04\textwidth,
    moredelim=[is][\color{gray}]{|}{|},
    #1
  }
}
{}
\DeclareMathOperator{\out}{out}
\newcommand{\cheq}{\mathrel{\check{=}}}
\usepackage{xspace}
\newcommand*{\ComplexityClass}[1]{\textsc{#1}\xspace}
\newcommand*{\NL}{\ComplexityClass{NL}}
\newcommand*{\LogSpace}{\ComplexityClass{LogSpace}}

\usepackage{hyperref}
\definecolor{darkblue}{rgb}{0,0.1,0.5}
\hypersetup{colorlinks,
  linkcolor=darkblue,
  anchorcolor=darkblue,
  citecolor=darkblue}

\usepackage[affil-it]{authblk}

\author{Jan Philipp Wächter\thanks{This research was supported by EPSRC project EP/Y008626/1 'Special Inverse Monoids: Geometry, Structure \& Algorithms'.}}
\affil{Department of Mathematics\\
  University of Manchester\\
  Oxford Road\\
  Manchester M13 9PL, UK}

\title{Context-Free Trees}

\begin{document}
  \maketitle
  \begin{abstract}
    Muller and Schupp introduced the concept of context-free graphs (originating from Cayley graphs of context-free groups).
    These graphs are always tree-like (i.\,e.\ quasi-isometric to a tree) and in this paper we investigate the subclass of bona fide context-free trees.
    We show that they have a finite-state description using multi-edge NFAs and that this specializes to certain partial DFAs in the case of deterministic graphs.
    We investigate this form of encoding algorithmically and show that the isomorphism problem for deterministic context-free trees is $\NL$-complete in the rooted and the non-rooted case.
      %20M18 inverse semigroups
      %20F65 geomertic group theory
  \end{abstract}
  
  \begin{section}{Introduction}
    The concept of context-free graphs was introduced by Muller and Schupp \cite{muller1985theory}.
    The motivation for considering these graphs came from studying the Cayley graphs of groups whose word problem is a context-free language (simply known as \emph{context-free groups}).
    In a celebrated result, the same authors had previously (almost) shown that the context-free groups are precisely the virtually free groups \cite{muller1983groups} (see also \cite{diekert2017contextfree}).
    Their proof originally required an additional hypothesis, which was later shown to be always satisfied by Dunwoody \cite{dunwoody1985accessibility}.
    Many algorithmic problems over groups are directly linked to properties of the Cayley graph (or related graphs), which motivates the algorithmic study of context-free graphs.
    Here, already Muller and Schupp gave a strong result based on Rabin's Tree Theorem \cite{rabin1969decidability}: the monadic second order theory of any context-free graph is decidable.
    Later, Kuske and Lohrey \cite{kuske2005logical} basically showed the converse, giving a strong link between context-free graphs and decidable monadic second order theory.
    
    Despite its strength, this does not suffice to answer questions on relations \emph{between} context-free graphs (for example whether they are isomorphic).
    Such questions arise in another branch of combinatorial algebra: the study of inverse semigroups and monoids, where the role of the Cayley graph is taken over by what is called the \emph{Schützenberger graphs}, around which a very rich theory has evolved; in particular the seminal work of Stephen \cite{stephen1990presentations} is noteworthy here.
    The fundamental difference to the group case is that we now have to deal with multiple Schützenberger graphs instead of a single Cayley graph.
    Accordingly, we have to answer questions on multiple graphs.
    Of course, these graphs will not be context-free in general (as they are not even in the group case).
    However, for finitely presented \emph{tree-like} inverse monoids, they are \cite{gray2022algorithmic}, which brings us back to the need for algorithms on (multiple) such graphs.
    
    While much of the existing work on context-free graphs (see e.\,g.\ \cite{ceccherini2012contextfree, kuske2005logical, lindorfer2020language, pelecq1996automorphism, rodaro2024generalizations}) is motivated by their mentioned origin and applications in algorithmic (semi)group theory, they are also objects worthwhile to be studied in their own right.
    In fact, while the current paper certainly also draws its motivations from algorithmic (inverse semi)group theory, it can be understood completely independently and without a deep background in that area.
    
    The first obstacle in the algorithmic study of context-free graphs is how to finitely encode them as inputs to algorithms.
    Muller and Schupp's original definition is based on the end-cones of the graphs or, more loosely speaking, their behavior \enquote{at infinity}.
    While this does not seem to point towards a finite description of these graphs, they showed that context-free graphs also arise as the configuration graphs of push-down automata (PDAs) (see \cite{muller1985theory} directly for more details).
    However, PDAs are algorithmically already quite challenging and a simpler description would be more useful.
    
    In this paper, we will concentrate on the subclass of context-free \emph{trees}.
    This is motivated from two directions:
    First, such trees arise as Schützenberger graphs for example in the form of finite \emph{Munn trees} in free inverse monoids \cite{munn1974free} but also as infinite trees in some other finitely presented inverse monoids \cite{margolis1993inverse} (with \cite{gray2022algorithmic}).
    Second, all context-free graphs are tree-like, in the sense that they are quasi-isometric to trees and, equivalently, admit strong tree decompositions \cite{antolin2011cayley, gray2022algorithmic}.
    In some sense, context-free graphs seem to be \enquote{dominated} by this underlying tree and it seems useful to investigate this further.
    
    We show that context-free trees have a finite description using a multi-edge nondeterministic finite automaton.
    The finite-state nature of this description makes it algorithmically approachable and there are arguments to be made that it is, in fact, the natural encoding of such objects.
    We further show that, for deterministic context-free trees, this description is a partial deterministic finite automaton (pDFA) (in which we do not find any path labeled by $aa^{-1}$).
    The deterministic case is highly interesting with respect to the above algebraic setting because the context-free graphs arising as Cayley or Schützenberger graphs are always deterministic.
    We show that, using this pDFA description, the isomorphism problem for deterministic context-free trees is decidable and, in fact, $\NL$-complete.
    This holds true in the case of rooted isomorphisms but also in the more complicated non-rooted case.
    
    It seems likely that the results of this paper could serve as a starting point for investigating finite-state descriptions of more general context-free graphs and to consider further algorithmic problems beyond the isomorphism problem.
    Another direction where the current results could serve as a stepping stone is computing descriptions of the isomorphisms (and, thus, also automorphisms) of a context-free tree (or more general graph).
    Such a description would be useful as the automorphism groups of the Schützenberger graphs are precisely the maximal subgroups of an inverse monoid \cite{stephen1990presentations} (an area which is actively researched, see e.\,g.\ \cite{gray2025maximal}).
  \end{section}
  \begin{section}{Preliminaries}
    \paragraph*{Complexity Theory.}
    We need fundamental notions of complexity-theory.
    In particular, we need the complexity class $\NL$ of problems solvable in nondeterministic logarithmic space.
    We also need the notion of being $\NL$-complete (where we use many-one \LogSpace-reductions computed by \LogSpace-transducers).
    For background information on these, the reader is referred to any standard textbook on the topic (e.\,g.\ \cite{papadimitriou97computational}).
    
    \paragraph*{Alphabets and Words.}
    An \emph{alphabet} is a set $A$ of \emph{letters.}
    An \emph{involutive alphabet} is an alphabet $A$ equipped with an involution $\cdot^{-1} \colon A \to A$ (where $a^{-1} = a$ is possible).
    To every (ordinary) alphabet $A$, we may attach the involutive alphabet $A^{\pm 1} = A \uplus A^{-1}$ where $A^{-1} = \{ a^{-1} \mid a \in A \}$ is a disjoint copy with the involution $a^{-1} = a$ and $(a^{-1})^{-1} = a$.
    
    A \emph{word} $w$ over an alphabet $A$ (involutive or not) is a finite sequence $w = a_1 \dots a_\ell$ of letters $a_1, \dots, a_n \in A$.
    Its \emph{length} is $\ell$ and the \emph{empty} word (i.\,e.\ the unique word of length $0$) is denoted by $\varepsilon$.
    The set of all words over $A$ is denoted by $A^*$, which has a natural operation of concatenating words.
    We write $A^\ell$ for the set of words over $A$ of length exactly $\ell$ and natural variations of this notation such as writing $A^{\leq \ell}$ for the set of words of length at most $\ell$.
    A word $u$ is a \emph{prefix} of another word $w$ if there is some $v$ such that $w = uv$.
    Finally, a (formal) \emph{language} is a subset of $A^*$ for some $A$.
    
    \paragraph*{Graphs.}
    Let $A$ be an alphabet.
    An \emph{$A$-graph} is a directed, $A$-edge labeled graph $\Gamma = (V, E)$ with node set $V$ and edge set $E \subseteq V \times A \times V$.
    In the context of an edge, we write $\aedge{u}{a}{v}$ but also $\laedge{v}{a}{u}$ for the element $(u, a, v) \in V \times A \times V$.
    Such an edge \emph{starts} in $u$, \emph{ends} in $v$ and is \emph{labeled} by $a$.
    The \emph{out-degree} (\emph{in-degree}) of a vertex is the cardinality of the set of edges that start (end) in it.
    We say the out-degree (in-degree) of an $A$-graph is \emph{uniformly bounded} if there is a finite number bounding the out-degree (in-degree) of every vertex.
    
    A \emph{path} of \emph{length} $\ell \in \mathbb{N}$ in the $A$-graph $\Gamma = (V, E)$ is a finite sequence $e_1 \dots e_\ell$ of edges $e_1, \dots, e_\ell \in E$ such that the end vertex of $e_{i}$ is the start vertex of $e_{i + 1}$.
    The path \emph{starts} in the same vertex $v_0$ as $e_1$ and \emph{ends} in the same vertex $v_\ell$ as $e_\ell$; it is a path \emph{from} $v_0$ \emph{to} $v_\ell$.
    If the edges are labeled by $a_1, \dots, a_\ell$, respectively, the \emph{label} of the path is $a_1 \dots a_\ell \in A^*$.
    
    An \emph{isomorphism} between two $A$-graphs (or an \emph{$A$-isomorphism} if we want to stress that we are considering $A$-graphs) $\Gamma = (V, E)$ and $\Delta = (W, F)$ is a bijective function $\varphi\colon V \to W$ such that we have an edge $\aedge{v_1}{a}{v_2} \in E$ if and only if there is an edge $\aedge{\varphi(v_1)}{a}{\varphi(v_2)} \in F$.
    
    Sometimes we denote a special node as the \emph{root} of an $A$-graph and speak of a \emph{rooted} $A$-graph. An isomorphism between rooted $A$-graphs must map the root to the root of the other graph.
    
    We will also consider \emph{node-labeled} $A$-graphs, which we equip with a function $q\colon V \to Q$ from the node set to some alphabet $Q$.
    If $\Gamma = (V, E)$ is an $A$-graph with node labeling $p\colon V \to P$ and $\Delta = (W, F)$ is an $A$-graph with node labeling $q\colon W \to Q$, then $\Gamma$ and $\Delta$ are \emph{isomorphic} as node-labeled $A$-graphs if there is an $A$-isomorphism $\iota\colon V \to W$ and a bijection $\beta\colon P \to Q$ such that $\beta(p(v)) = q(\iota(v))$.

    \paragraph*{Involutive and Context-Free Graphs.}
    An \emph{involutive} graph is an $A$-graph $\Gamma = (V, E)$ for an involutive alphabet $A$ where we have $\aedge{u}{a}{v} \in E \iff \laedge{u}{a^{-1}}{v} \in E$.
    Thus, in an involutive graph, the out- and the in-degrees coincide, which allows us to speak of the \emph{degree} of a node.
    We may extend the involution to the edge set: We call the edge $e^{-1} = \laedge{u}{a^{-1}}{v}$ the \emph{inverse} of the edge $e = \aedge{u}{a}{v}$.
    
    Generally, we may define the \emph{involutive closure} $\Gamma^{\pm 1}$ of an $A$-graph $\Gamma$: If $A$ is not an involutive alphabet, we consider $\Gamma$ over the alphabet $A^{\pm 1}$ instead. On an involutive alphabet, the involutive closure of $\Gamma = (V, E)$ is then the $A$-graph $\Gamma^{\pm 1} = (V, E \cup E^{-1}$) with $E^{-1} = \{ e^{-1} \mid e \in E \}$.
    Typically, when depicting an involutive graph, we only draw one edge for every involutive pair.
    
    The involution on the edges even extends into one on the paths:
    the \emph{inverse} of a path $\pi = e_1 \dots e_\ell$ from $u$ to $v$ is the path $\pi^{-1} = e_\ell^{-1} \dots e_1^{-1}$ from $v$ to $u$ of the same length.
    A path $\pi$ is \emph{reduced} if it cannot be written as $\pi = \varrho_1 e e^{-1} \varrho_2$ for any edge $e$ and paths $\varrho_1, \varrho_2$.

    An involutive graph is \emph{connected} if, for any two nodes $u$ and $v$, there is a path from $u$ to $v$.
    This allows us to define the \emph{distance} metric between the nodes of a connected involutive graph as the length of a shortest path from one to the other.
    
    If we are dealing with such a rooted involutive $A$-graph $\Gamma$, we let the \emph{level} $\ell(v)$ of node $v$ be its distance from the root.
    Note that we may only have edges between nodes whose levels differ by at most $1$ (i.\,e.\ edges to nodes on the previous, same or next level).
    The \emph{disc} $\Gamma^{\leq n}$ of \emph{radius} $n \in \mathbb{N}$ (around the root) consists of the nodes of level at most $n$.
    We may also include the edges between those nodes and see it as a subgraph.
    Even though $\Gamma$ is connected, the graph $\Gamma_{n + 1} = \Gamma \setminus \Gamma^{\leq n}$ may not be connected.
    For consistency, we also write $\Gamma_0$ for $\Gamma$ itself.
    For a node $v$ at level $n$, we use $\Gamma(v)$ to denote the connected component of $\Gamma_n$ which contains $v$ (i.\,e.\ $\Gamma(v)$ contains precisely those nodes reachable from $v$ by using only edges between nodes of level $n$ or larger).
    Such a connected component $\Gamma(v)$ is called an \emph{end-cone} of $\Gamma$ and the nodes of level precisely $n$ in it are its \emph{frontier points}.
    An \emph{end-isomorphism} is an $A$-isomorphism between two end-cones $\Gamma(u)$ and $\Gamma(v)$ that maps frontier points of $\Gamma(u)$ on frontier points of $\Gamma(v)$.
    
    A \emph{context-free} graph is a connected involutive $A$-graph for a finite involutive alphabet $A$ with uniformly bounded degree such that there is a root resulting in finitely many end-isomorphism classes.
    It turns out that this property does not depend on the chosen root \cite[2.7.~Corollary]{muller1985theory}; we may, thus, use any vertex as the root of a context-free graph.
    
    \paragraph*{Trees.}
    A rooted involutive graph is a \emph{tree} if every vertex is the end of a unique reduced path starting at the root.
    We call such a graph an \emph{involutive $A$-tree} or simply an \emph{involutive tree}.
    This property is, in fact, independent of the root and we may speak of involutive trees without having a denoted root.
    An $A$-graph is a \emph{tree} if its involutive closure is.
    
    In a tree, we may not have edges between nodes of the same level (as those edges would yield multiple reduced paths from the root).
    Thus, edges only go to the previous or next level.
    Similarly, we have:
    \begin{fact}
      In a rooted involutive tree $\Gamma$, the node $v$ is the only frontier point of the end-cone $\Gamma(v)$.
    \end{fact}
    \begin{proof}
      Suppose we have two nodes $u$ and $v$ of the same level $n$ connected by a path $\pi_{uv}$ visiting only nodes of level $n$ or larger.
      Without loss of generality, we may assume that this path is reduced.
      Additionally, there is are two unique reduced paths $\pi_u$ and $\pi_v$ from the root to $u$ and $v$, respectively.
      Both of them must have length $n$.
      The last edge of $\pi_u$ connects $u$ to a node on level $n - 1$ and must, thus, differ from the first edge of $\pi_{uv}$.
      In other words, $\pi_u \pi_{uv}$ is another reduced path from $u$ to $v$ which differs from $\pi_v$ (as it has different length).
    \end{proof}
    \noindent
    An immediate consequence of the previous fact is that we may consider end-cones in involutive trees as rooted graphs and that end-isomorphisms are the same thing as rooted isomorphisms.
  \end{section}
  
  \begin{section}{Context-Free Trees}
    \begin{subsection}{The General Case}
      \paragraph*{Multi-Edge Nondeterministic Finite Automata.}
      A \emph{multi-edge nondeterministic finite automaton} (or \emph{mNFA} for short) is a finite $A$-edge-labeled multigraph. Formally, we consider it as a tuple $\mathcal{A} = (Q, A, T, \alpha, \lambda, \omega)$ where $Q$ is a finite set of \emph{states}, $A$ is a finite alphabet, $T$ is a finite set of \emph{transitions}, $\alpha, \omega\colon T \to Q$ are functions yielding the \emph{start} $\alpha(\tau)$ and the \emph{end} $\omega(\tau)$ of a transition $\tau \in T$ and $\lambda\colon T \to A$ is a function describing the \emph{label} $\lambda(\tau)$ of a transition $\tau \in T$.
      
      For a transition $\tau$ with $\alpha(\tau) = p$, $\lambda(\tau) = a$ and $\omega(\tau) = q$, we also use the graphical notation $\tau = \aedge{p}{a}{q}$. In addition, we use the common way of drawing finite automata but note that we may have multiple $a$-labeled edges between any two states.
      \begin{figure}\centering
        \begin{minipage}[t]{\dimexpr0.5\linewidth-0.25cm}
          \vspace{0pt}%
          \begin{subfigure}{\linewidth}\centering
            \begin{tikzpicture}[auto, shorten >=1pt, >=latex]
              \node[state] (q) {$p$};
              \draw[->] (q) edge[loop left] node {$a$} (q)
                            edge[loop right] node {$a$} (q);
            \end{tikzpicture}
            \caption{A very simple mNFA with one state and two self-loops $\tau_0$ and $\tau_1$.}
          \end{subfigure}\\%
          \begin{subfigure}{\linewidth}\centering
            \begin{tikzpicture}[
                shorten >=1pt, >=latex,
                level distance=1.5cm, auto,
                level 1/.style={sibling distance=3cm},
                level 2/.style={sibling distance=1.5cm},
                level 3/.style={sibling distance=1cm, level distance=1cm},
                level 4/.style={sibling distance=0.5cm}]
              \node {$\varepsilon$}
                child {node {$\tau_0$}
                  child {node {$\tau_0 \tau_0$}
                    child {
                      node[inner sep=0pt] {$\vdots$}
                      edge from parent[->] node[swap] {$a$}
                    }
                    child {
                      node[inner sep=0pt] {$\vdots$}
                      edge from parent[{latex[fill=lightgray]}->] node {$a$}
                    }
                    edge from parent[->] node[swap] {$a$}
                  }
                  child {node {$\tau_0 \tau_1$}
                    child {
                      node[inner sep=0pt] {$\vdots$}
                      edge from parent[->] node[swap] {$a$}
                    }
                    child {
                      node[inner sep=0pt] {$\vdots$}
                      edge from parent[{latex[fill=lightgray]}->] node {$a$}
                    }
                    edge from parent[{latex[fill=lightgray]}->] node {$a$}
                  }
                  edge from parent[->] node[swap] {$a$}
                }
               child {node {$\tau_1$}
                 child {node {$\tau_1 \tau_0$}
                   child {
                     node[inner sep=0pt] {$\vdots$}
                     edge from parent[->] node[swap] {$a$}
                   }
                   child {
                     node[inner sep=0pt] {$\vdots$}
                     edge from parent[{latex[fill=lightgray]}->] node {$a$}
                   }
                   edge from parent[->] node[swap] {$a$}
                 }
                 child {node {$\tau_1 \tau_1$}
                   child {
                     node[inner sep=0pt] {$\vdots$}
                     edge from parent[->] node[swap] {$a$}
                   }
                   child {
                     node[inner sep=0pt] {$\vdots$}
                     edge from parent[{latex[fill=lightgray]}->] node {$a$}
                   }
                   edge from parent[{latex[fill=lightgray]}->] node {$a$}
                 }
                 edge from parent[{latex[fill=lightgray]}->] node {$a$}
               };
            \end{tikzpicture}
            \caption{The associate tree $\Gamma(p)$. All nodes are labeled with $p$. Only one edge per involutive pair is drawn; oriented in the black direction.}\label{sfig:generatedTree}
          \end{subfigure}
          \caption{Example of an mNFA with an associated tree.}\label{fig:ex1}
        \end{minipage}\hfill%
        \begin{minipage}[t]{\dimexpr0.5\linewidth-0.25cm}\centering
          \vspace{0pt}%
          \begin{tikzpicture}[auto, shorten >=1pt, >=latex]
            \node[state] (q) {$p$};
            \draw[->] (q) edge[loop left] node {$a$} (q)
                          edge[loop right] node {$a^{-1}$} (q);
          \end{tikzpicture}
          \caption{An mNFA for \autoref{sfig:generatedTree} with the gray orientation.}\label{fig:exInvolutive}
          \nextfloat
          \begin{subfigure}{\linewidth}\centering
            \begin{tikzpicture}[auto, shorten >=1pt, >=latex]
              \node[state] (p) {$p$};
              \node[state, right=of p] (q) {$q$};
              \draw[->] (p) edge[loop left] node {$a$} (p)
                            edge node {$b$} (q)
                        (q) edge[loop right] node {$b$} (q)
                        ;
            \end{tikzpicture}
            \caption{An mNFA with two states.}
          \end{subfigure}\\%
          \begin{subfigure}{\linewidth}\centering
            \begin{tikzpicture}[auto, on grid, shorten >=1pt, >=latex]
              \node[fill, circle, inner sep=1.5pt] (e) {};
              \node[fill, circle, inner sep=1.5pt, right=of e] (b) {};
              \node[right=of b] (bb) {$\ldots$};
              
              \node[fill, circle, inner sep=1.5pt, below=0.75cm of e] (a) {};
              \node[fill, circle, inner sep=1.5pt, right=of a] (ab) {};
              \node[right=of ab] (abb) {$\ldots$};
              
              \node[fill, circle, inner sep=1.5pt, below=0.75cm of a] (aa) {};
              \node[fill, circle, inner sep=1.5pt, right=of aa] (aab) {};
              \node[right=of aab] (aabb) {$\ldots$};
              
              \node[below=of aa, yshift=1ex, inner sep=0pt] (aaa) {$\vdots$};
              
              \draw[->] (e) edge node[left] {$a$} (a)
                            edge node {$b$} (b)
                        (b) edge node {$b$} (bb)
              ;
              \draw[->] (a) edge node[left] {$a$} (aa)
                            edge node {$b$} (ab)
                        (ab) edge node {$b$} (abb)
              ;
              \draw[->] (aa) edge node[left] {$a$} (aaa)
                             edge node {$b$} (aab)
                        (aab) edge node {$b$} (aabb)
              ;
            \end{tikzpicture}
            \caption{The graph $\Gamma(p)$. The nodes on the left are labeled $p$, all other nodes are labeled $q$.}
          \end{subfigure}
          \caption{Another example of an mNFA with the associated tree.}\label{fig:ex2}
        \end{minipage}
      \end{figure}
      
      A \emph{run} of the mNFA $\mathcal{A}$ is a sequence $\varrho = \tau_1 \dots \tau_\ell$ of transitions $\tau_1, \dots, \tau_\ell \in T$ with $\omega(\tau_i) = \alpha(\tau_{i + 1})$ (for all $i \leq i < \ell$).
      It \emph{starts} in $\alpha(\varrho) = \alpha(\tau_1)$, \emph{ends} in $\omega(\varrho) = \omega(\tau_\ell)$, has \emph{length} $\ell$ and is \emph{labeled} with $\lambda(\varrho) = \lambda(\tau_1) \dots \lambda(\tau_\ell)$.
      A state $q$ is \emph{reachable} or \emph{accessible} form a state $p$ if there is a run starting in $p$ and ending in $q$.
      A state is called a \emph{root} of $\mathcal{A}$ if all states are reachable from it.
      
      To every state $p \in Q$ of an mNFA $\mathcal{A} = (Q, A, T, \alpha, \lambda, \omega)$, we may associate the tree of computations/runs of $\mathcal{A}$ starting in $p$.
      Formally, we consider the $A$-graph whose node set $V$ is the set of runs of $\mathcal{A}$ starting in $p$.
      We have an edge from $\varrho$ to $\pi$ whenever $\pi = \varrho \tau$ for some transition $\tau \in T$.
      It is labeled by $\lambda(\tau)$.
      We denote the empty path $\varepsilon$ as the root of this graph and observe that all edges are oriented away from this root.
      We may assume $A$ to be an involutive alphabet (by passing to $A^{\pm 1}$ if necessary) and obtain an involutive $A$-graph by adding all inverse edges.
      This involutive $A$-graph is easily verified to be a tree and we denote it by $\Gamma(p)$.
      The degree in $\Gamma(p)$ is uniformly bounded by $|T|$.
      We may endow this graph with a node labeling: Let $\varrho \in T^*$ be a run of $\mathcal{A}$ starting in $p$ (i.\,e.\ a node of $\Gamma(p)$). This run ends in some state $q$ and we use this state as the label $q(\varrho) = q$ of $\varrho$ in $\Gamma(p)$.
      The reader may find examples of two mNFAs with a corresponding associated tree in \autoref{fig:ex1}, \autoref{fig:exInvolutive} and \autoref{fig:ex2}.
      
      Sometimes it is convenient to use a recursive description of $\Gamma(p)$ instead.
      For an $A$-graph $\Gamma = (V, E)$ with $V \subseteq T^*$ and $\tau \in T$, let $\tau\Gamma$ be the $A$-graph with node set $\tau V = \{ \tau v \mid v \in V \subseteq T^* \}$ and edge set $\tau E = \{ \aedge{\tau u}{a}{\tau v} \mid \aedge{u}{a}{v} \in E \}$. If $\Gamma$ has a root $r$, we denote $\tau r$ as the root of $\tau\Gamma$ and, for a node labeling $q\colon V \to Q$, we let $\tau q\colon \tau V \to Q$ with $\tau q(\tau v) = q(v)$.
      Clearly, $\Gamma$ is isomorphic to $\tau\Gamma$ as a (rooted/node-labeled) $A$-graph (via the $A$-isomorphism $v \mapsto \tau v$).
      
      Let $\{ \tau_1, \dots, \tau_n \}$ be the set of transition starting in $p$.
      We may write $\tau_i = \aedge{p}{a_i}{q_i}$ (for $i = 1, \dots, n$) and observe that $\Gamma(p)$ is the $A$-graph arising from the union $\bigcup_{i = 1}^n \tau_i \Gamma(q_i)$ with the additional node $\varepsilon$ and the additional edges $\{ \aedge{\varepsilon}{a_i}{\tau_i} \mid i = 1, \dots, n \}$ and their inverses (where $\varepsilon$ is labeled with $p$).
      
      \paragraph*{Regular Trees.}
      Using mNFAs, we may define what a regular tree is.
      \begin{definition}
        A rooted, node-labeled involutive $A$-tree is \emph{regular} if it is isomorphic to $\Gamma(p)$ (as a rooted, node-labeled $A$-graph) for some state $p$ of an mNFA with alphabet $A$.
        A rooted involutive $A$-tree is \emph{regular} if it is regular for some node labeling and
        an involutive $A$-tree is \emph{regular} if it is regular with some vertex chosen as the root.
      \end{definition}
      \enlargethispage{\baselineskip}
      It turns out that regularity does not depend on the chosen root, which allows us to choose any node as the root for a regular involutive $A$-tree:
      \begin{proposition}\label{prop:regularityIsRootless}
        Let $\Gamma$ be an involutive $A$-tree and $u$ and $v$ be two of its nodes.
        Then, $\Gamma$ is regular with respect to $u$ as the root if and only if it is regular with $v$ as the root.
      \end{proposition}
      \begin{proof}
        It suffices to show the statement for nodes $u$ and $v$ of distance $1$ (as the rest then follows by induction).
        
        We may assume that $\Gamma$ is equal to $\Gamma(p)$ for some state $p$ of an mNFA $\mathcal{A} = (Q, A, T, \alpha, \lambda,\allowbreak \omega)$.
        Let $\sigma_0$ be a neighbor of the root $\varepsilon$ in $\Gamma(p)$; i.\,e.\ we have $T \ni \sigma_0 = \aedge{p}{a_0}{q_0}$.
        We need to show that, if we change the root to be $\sigma_0$, then the resulting graph is isomorphic to $\Gamma(q')$ as a rooted $A$-graph for some state $q'$ of a (potentially different) mNFA $\mathcal{A}'$ (which also yields a new node labeling).
        
        Let $\sigma_0, \dots, \sigma_m$ denote the transitions of $\mathcal{A}$ starting at $p$ and write $\sigma_i = \aedge{p}{a_i}{q_i}$ (for $i = 0, \dots, m$).
        Furthermore, let $\tau_1, \dots, \tau_n$ denote the transitions starting in $q_0$ and write $\tau_i = \aedge{q_0}{b_i}{r_i}$ (for $i = 1, \dots, m$).
        
        Using the recursive description of $\Gamma(p)$, we may now depict it in the following way (where we only draw one edge for each involutive edge pair):
        
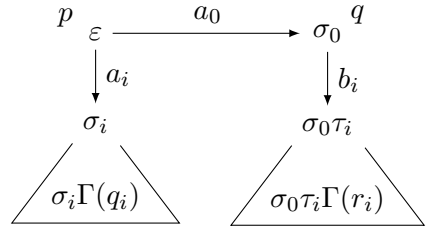
\begin{center}
          \begin{tikzpicture}[auto, shorten >=1pt, >=latex]
            \node (e) {$\varepsilon$};
            \node[right=2.5cm of e] (s0) {$\sigma_0$};
            \node[above left=0.25ex and 1ex of e.base] {$p$};
            \node[above right=0.25ex and 1ex of s0.base] {$q$};
            
            \node[below=0.75cm of e] (si) {$\sigma_i$};
            \node[below=0.75cm of s0] (ti) {$\sigma_0 \tau_i$};
            
            \draw[->] (e) edge node {$a_0$} (s0)
                          edge node {$a_i$} (si);
            \draw[->] (s0) edge node {$b_i$} (ti);
            
            \node[below=1cm of si.base, anchor=base] (gqi) {$\sigma_i \Gamma(q_i)$};
            \draw (si.south west) -- ($(gqi.base west)+(-1em,-1.5ex)$) -- ($(gqi.base east)+(1em,-1.5ex)$) -- (si.south east);
            
            \node[below=1cm of ti.base, anchor=base] (gri) {$\sigma_0 \tau_i \Gamma(r_i)$};
            \draw (ti.south west) -- ($(gri.base west)+(-1em,-1.5ex)$) -- ($(gri.base east)+(1em,-1.5ex)$) -- (ti.south east);
          \end{tikzpicture}
        \end{center}
        
        In order to construct $\mathcal{A}'$, we start with $\mathcal{A}$ and add new states $p'$ and $q'$ as well as the following (pairwise disjoint) transitions:
        \begin{align*}
          \sigma_i' &= \aedge{p'}{a_i}{q_i} \text{ for } i = 1, \dots, m \text{ (but not for $i = 0$),} \\
          \tau_0' &= \aedge{q'}{a_0^{-1}}{p'} \text{ and} \\
          \tau_i' &= \aedge{q'}{b_i}{r_i} \text{ for } i = 1, \dots, n
        \end{align*}
        Effectively, we have duplicated the state $p$ into $p'$ with almost all of its transitions but we have excluded $\sigma_0$. Similarly, we have duplicated $q_0$ into $q'$ and added an additional $a_0^{-1}$-transition to $p'$.
        
        We now obtain the following $\Gamma(q')$ (again only drawing one edge per involutive pair):%
        
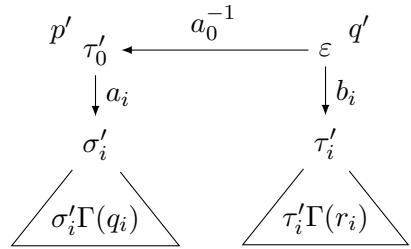
\begin{center}
          \begin{tikzpicture}[auto, shorten >=1pt, >=latex]
            \node (e) {$\tau_0'$};
            \node[right=2.5cm of e] (s0) {$\varepsilon$};
            \node[above left=0.25ex and 1ex of e.base] {$p'$};
            \node[above right=0.25ex and 1ex of s0.base] {$q'$};
            
            \node[below=1.25cm of e.base, anchor=base] (si) {$\sigma_i'$};
            \node[below=1.25cm of s0.base, anchor=base] (ti) {$\tau_i'$};
            
            \draw[->] (s0) edge node[above] {$a_0^{-1}$} (e)
                      (e) edge node {$a_i$} (si);
            \draw[->] (s0) edge node {$b_i$} (ti);
            
            \node[below=1cm of si.base, anchor=base] (gqi) {$\sigma_i' \Gamma(q_i)$};
            \draw (si.south west) -- ($(gqi.base west)+(-1em,-1.5ex)$) -- ($(gqi.base east)+(1em,-1.5ex)$) -- (si.south east);
            
            \node[below=1cm of ti.base, anchor=base] (gri) {$\tau_i' \Gamma(r_i)$};
            \draw (ti.south west) -- ($(gri.base west)+(-1em,-1.5ex)$) -- ($(gri.base east)+(1em,-1.5ex)$) -- (ti.south east);
          \end{tikzpicture}
        \end{center}
        As an involutive rooted $A$-graph, this is isomorphic to $\Gamma(p)$ if we change the root there to $\sigma_0$.
        Thus, $\Gamma(p)$ remains regular after changing the root to a neighbor.
      \end{proof}
      
      In fact, we obtain that, for trees, being regular and being context-free coincide:%
      \begin{theorem}\label{thm:contextFreeIsRegular}%
        An involutive tree is regular if and only if it is context-free.
      \end{theorem}
      \begin{proof}
        Let $\Gamma$ be an involutive $A$-tree. As a tree, $\Gamma$ is connected.
        
        First, we assume that it is regular, which means that, for some root and some node labeling, it is isomorphic to $\Gamma(q_0)$ (as a rooted, node-labeled $A$-graph) for some state $q_0$ of an mNFA with alphabet $A$.
        It suffices to show that $\Gamma(q_0)$ is context-free.
        As the alphabet of an mNFA, $A$ is finite and the degree of any node in $\Gamma(q_0)$ is trivially bounded by the (finite) number of transitions in the mNFA.
        Thus, it remains to show that $\Gamma(q_0)$ has finitely many end-isomorphism classes.
        However, the end-isomorphism class of $\Gamma(u)$ for a node $u$ is uniquely determined by the node label of $u$.
        More specifically, we claim that, if $u$ is labeled by the state $p$, then $\Gamma(u) = u \Gamma(p)$ and $u$ is the only frontier point of this end-cone.
        Since $u \Gamma(p)$ is isomorphic to $\Gamma(p)$ as a rooted graph, this yields an end-isomorphism.
        
        The claim itself can be shown using an induction on the recursive description of $\Gamma(q_0)$.
        The statement is true for the root $u = \varepsilon$ (which is labeled by $q_0$) by definition.
        Suppose we already have established $\Gamma(u) = u\Gamma(p)$ where $p$ is the label of $u$ and $u$ is the only frontier point.
        We need to show $\Gamma(u\tau) = u\tau \Gamma(q)$ for any transition $\tau = \aedge{p}{a}{q}$ (where $a \in A$).
        We obtain $\Gamma(u\tau)$ from $\Gamma(u)$ by removing all nodes at the same level as $u$ from it (these are the frontier points; i.\,e.\ we only remove $u$) and selecting the connected component containing $u\tau$.
        Since $v \mapsto uv$ is an $A$-isomorphism from $\Gamma(p)$ to $\Gamma(u) = u\Gamma(p)$, we may apply the recursive description of $\Gamma(p)$.
        Let $\tau_1, \dots, \tau_d$ be the transitions starting in $p$ and write $\tau_i = \aedge{p}{a_i}{q_i}$ (for $i = 1, \dots, d$).
        We obtain that $\Gamma(u)$ is equal to $\bigcup_{i = 1}^d u \tau_i \Gamma(q_i)$ with the additional node $u$ connected to each $u \tau_i$ with an $a_i$-labeled edge (and its inverse).
        In particular, we have all paths staring and ending in different $u \tau_i \Gamma(q_i)$ must pass through $u$ and the roots $u \tau_i$ are the only nodes directly connected to $u$ (i.\,e.\ the nodes on the next level).
        This shows the claim.
        
        For the converse direction, let $\Gamma$ be context-free.
        In particular, $A$ is finite and there is a finite number $D$ bounding the degree at any node.
        We may also choose a root $r$ with finitely many end-isomorphism classes.
        For these classes, we may choose $\Gamma(r_1), \dots, \Gamma(r_n)$ as representatives.
        Without loss of generality, we assume $r = r_1$.
        We endow $\Gamma$ with the node labeling mapping a node $v$ to $i \in \{ 1, \dots, n \}$ where $\Gamma(v)$ is end-isomorphic to $\Gamma(r_i)$.
        
        We construct an mNFA over the alphabet $A$ with state set $\{ 1, \dots, n \}$ such that $\Gamma(i)$ is isomorphic to $\Gamma(r_i)$ as a rooted, node-labeled $A$-graph (i.\,e.\ end-isomorphic).
        In particular, $\Gamma(1)$ is isomorphic to $\Gamma(r_1) = \Gamma(r) = \Gamma$, which is, thus, regular.
        
        For the transitions, consider $\Gamma(r_i)$ for some $i \in \{ 1, \dots, n \}$.
        Let $e_1, \dots, e_d$ with $e_k = \aedge{r_i}{a_k}{v_k}$ for $d \leq D$ be the edges starting in $r_i$ inside $\Gamma(r_i)$.
        Furthermore, let $j_k$ be the label of $v_k$ for each $k = 1, \dots, d$ (i.\,e.\ $\Gamma(v_k)$ is end-isomorphic to $\Gamma(r_{j_k})$).
        For each such edge, we add a transition $\tau_k = \aedge{i}{a_k}{j_k}$ to our mNFA (potentially creating multi-edges if there is more than one edge labeled by some $a$ to a node labeled by some $j$).
        
        The thus constructed mNFA has the above property (we may show by induction on $\ell$ that, for all $\ell \geq 0$ and $i \in \{ 1, \dots, n \}$, the disc of radius $\ell$ in $\Gamma(r_i)$ is isomorphic to $\Gamma^{\leq \ell}(i)$ using the construction of the mNFA and the recursive description of $\Gamma(i)$).
      \end{proof}
      \begin{remark}
        \autoref{thm:contextFreeIsRegular} states, in particular, that any context-free involutive tree is isomorphic to $\Gamma(p)$ for some state $p$ of an mNFA.
        This yields a finite description of the context-free tree where we might additionally assume that $p$ is a root of the mNFA (as we may drop any non-reachable states since this will not change $\Gamma(p)$).
        One might argue that this is indeed the \emph{natural} way to finitely describe a context-free tree and we will use this description of context-free involutive trees whenever we consider algorithms.
      \end{remark}
    \end{subsection}
    
    \begin{subsection}{The Deterministic Case}
      \paragraph*{Deterministic $A$-Graphs and pDFAs.}
      An $A$-graph $\Gamma = (V, E)$ is \emph{deterministic} if any pair $\aedge{u}{a}{v} \in E$ and $\aedge{u}{a}{v'} \in E$ of edges implies $v = v'$.
      A \emph{partial deterministic finite automaton} (or \emph{pDFA} for short) is a finite deterministic $A$-graph for some finite alphabet $A$; see \autoref{ex:pDFAs} below.
      The nodes of a pDFA are usually called \emph{states}, $A$ is called its alphabet and its edges are called \emph{transitions}.
      A finite path in $\mathcal{A}$ is, in the context of pDFAs, also called a \emph{run} of $\mathcal{A}$.
      
      For any state $p$ of a pDFA $\mathcal{D} = (Q, T)$ over $A$, we define $\out p = \{ a \mid \aedge{p}{a}{q} \in T \} \subseteq A$ as the set of letters readable from $p$ and let $p \cdot a \in Q$ be the unique state reached from $p$ by reading an $a \in \out p$.
      
      We may consider the pDFA $\mathcal{D}$ as a special case of an mNFA: we let $\mathcal{A} = (Q, A, T, \alpha, \lambda, \omega)$ with $\alpha(\aedge{p}{a}{q}) = p$, $\lambda(\aedge{p}{a}{q}) = a$ and $\omega(\aedge{p}{a}{q}) = q$.
      This yields a description of the mNFAs that belong to pDFAs:
      \begin{fact}\label{fct:pDFAandmNFA}
        An mNFA $\mathcal{A} = (Q, A, T, \alpha, \lambda, \omega)$ comes from a pDFA if and only if, for all $\tau, \tau' \in T$ with $\alpha(\tau) = \alpha(\tau')$ and $\lambda(\tau) = \lambda(\tau')$, we have $\tau = \tau'$.
      \end{fact}
      \begin{proof}
        The just defined interpretation of a pDFA as an mNFA clearly satisfies the required condition.
        For the other direction, $T$ clearly maps injectively into $Q \times A \times Q$ by $\tau \mapsto \aedge{\alpha(\tau)}{\lambda(\tau)}{\omega(\tau)}$.
        Let $T'$ be the image of this map and observe that $(Q, T')$ is not only an $A$-graph but even a deterministic one.
        This pDFA yields precisely the original mNFA (under the above interpretation).
      \end{proof}
      
      For any state $p \in Q$ of a pDFA (over $A$), we may define the \emph{language} $\mathscr{L}(p)$ of $p$ as
      \[
        \mathscr{L}(p) = \{ w \in A^* \mid w \text{ labels a run starting in } p \}
      \]
      and observe that we have:
      \begin{fact}
        $\displaystyle\mathscr{L}(p) = \{ \varepsilon \} \cup \bigcup_{a \in \out p} a\mathscr{L}(p \cdot a)$
      \end{fact}
      For general mNFAs, we had to consider all runs to describe the trees $\Gamma(p)$ belonging to its states.
      In the case of a pDFA, the transitions starting at a state $p$ are in one-to-one correspondence with $\out p$ and we may write $\Gamma(p)$ as the involutive closure of the union of the $a\Gamma(p \cdot a)$ for $a \in \out p$ with the additional node $\varepsilon$ (labeled by $p$) and the additional edges $\{ \aedge{\varepsilon}{a}{a} \mid a \in \out p \}$ (where $a\Gamma(p\cdot a)$ is defined analogously to $\tau\Gamma(q_i)$ in the mNFA case).
      In the thus arising tree, we have that any node is from $A^*$ and its level is given by its length as a word over $A$.
      Indeed, edges $\aedge{u}{a}{v}$ only exist between nodes $u$ and $v$ whose lengths differ by exactly one.
      The shorter one of the two is the \emph{parent} of the other one and the longer one is a \emph{child}.
      Iterating these, we obtain the notion of an \emph{ancestor} and an \emph{descendant}: $u$ is an ancestor of $v$ if and only if it is a prefix of $v$ and, conversely, this is the case if and only if $v$ is a descendant of $u$.
      The end-cone of $v$ in $\Gamma(p)$ is given by its nodes from $vA^*$.
      We will later write $\Gamma(p) \setminus vA^*$ for the graph arising from $\Gamma(p)$ by removing all nodes belonging to the end-cone of $v$ (and edges incident to them).
      
      From this, it is not difficult to observe that $\Gamma(p)$ is fully described by $\mathscr{L}(p)$:
      \begin{fact}\label{fct:isoIsLang}
        Let $p$ be a state of a pDFA $\mathcal{A}$ and let $q$ be a state of a pDFA $\mathcal{B}$.
        Then the rooted, involutive graphs $\Gamma(p)$ and $\Gamma(q)$ are isomorphic (not considering node labels) if and only if $\mathscr{L}(p) = \mathscr{L}(q)$.
      \end{fact}
      \begin{proof}
        We will write $\Gamma \cheq \Delta$ to indicate that $\Gamma$ and $\Delta$ are isomorphic as rooted, involutive graphs (without node labels).
        We further write $\Gamma^{\leq \ell}(p)$ ($\Gamma^{\leq \ell}(q)$) for the disc of radius $\ell$ of $\Gamma(p)$ (of $\Gamma(q)$) and show that $\Gamma^{\leq \ell}(p) \cheq \Gamma^{\leq \ell}(q)$ if and only if $\mathscr{L}(p) \cap A^{\leq \ell} = \mathscr{L}(q) \cap A^{\leq \ell}$ for all $\ell$ by induction.
        
        For $\ell = 0$, we have that $\Gamma^{\leq 0}(p)$ and $\Gamma^{\leq 0}(q)$ consist only of $p$ and $q$, respectively, and are, thus, isomorphic.
        Similarly, we have $\mathscr{L}(p) \cap A^{\leq 0} = \{ \varepsilon \} = \mathscr{L}(q) \cap A^{\leq 0}$.
        
        For the inductive step from $\ell$ to $\ell + 1$, we observe that $\Gamma^{\leq \ell + 1}(p)$ consists of the union $a\Gamma^{\leq \ell}(p \cdot a)$ for all $a \in \out p$ and the additional node $\varepsilon$ with the edges $\{ \aedge{\varepsilon}{a}{a} \mid a \in \out p \}$ (and that an analogous statement holds for $\Gamma^{\leq \ell + 1}(q)$).
        We, thus, have $\Gamma^{\leq \ell + 1}(p) \cheq \Gamma^{\leq \ell + 1}(q)$ if and only if $\out p = \out q$ and $a\Gamma^{\leq \ell}(p \cdot a) \cheq \Gamma^{\leq \ell}(p \cdot a) \cheq \Gamma^{\leq \ell}(q \cdot a) \cheq a\Gamma^{\leq \ell}(q \cdot a)$ for all $a \in \out p = \out q$.
        By induction, we have $\Gamma^{\leq \ell}(p \cdot a) \cheq \Gamma^{\leq \ell}(q \cdot a)$ if and only if $\mathscr{L}(p \cdot a) \cap A^{\leq \ell} = \mathscr{L}(q \cdot a) \cap A^{\leq \ell}$ (where $a \in \out p \cap \out q$), which is the case if and only if $a\mathscr{L}(p \cdot a) \cap A^{\leq \ell + 1} = a\mathscr{L}(q \cdot a) \cap A^{\leq \ell + 1}$.
        
        Summing this up, we have $\Gamma^{\leq \ell + 1}(p) \cheq \Gamma^{\leq \ell + 1}(q)$ if and only if $\out p = \out q$ and $a\mathscr{L}(p \cdot a) \cap A^{\leq \ell + 1} = a\mathscr{L}(q \cdot a) \cap A^{\leq \ell + 1}$ for all $a \in \out p = \out q$.
        
        If both conditions are true, this implies
        \begin{align*}
          \mathscr{L}(p) \cap A^{\leq \ell + 1} &=
          \{ \varepsilon \} \cup \bigcup_{a \in \out p} \left( a\mathscr{L}(p \cdot a) \cap A^{\leq \ell + 1} \right) \\
          &= \{ \varepsilon \} \cup \bigcup_{a \in \out q} \left( a\mathscr{L}(q \cdot a) \cap A^{\leq \ell + 1} \right) \\
          &= \mathscr{L}(q) \cap A^{\leq \ell + 1} \text{.}
        \end{align*}
        For the other direction assume $\mathscr{L}(p) \cap A^{\leq \ell + 1} = \mathscr{L}(q) \cap A^{\leq \ell + 1}$.
        This implies, in particular, $\out p = \mathscr{L}(p) \cap A = \mathscr{L}(q) \cap A = \out q$.
        We obtain $a\mathscr{L}(p \cdot a) \cap A^{\leq \ell + 1} = a\mathscr{L}(q \cdot a) \cap A^{\leq \ell + 1}$ for each $a \in \out p = \out q$ by intersecting $\mathscr{L}(p)$ and $\mathscr{L}(q)$ with $aA^{\leq \ell} \subseteq A^{\leq \ell + 1}$, respectively.
      \end{proof}
      
      \paragraph*{Reduced pDFAs.}
      One is tempted to believe that the graphs $\Gamma(p)$ for states $p$ of pDFAs belong precisely to the deterministic regular trees.
      The involutive nature of the $\Gamma(p)$ requires an additional property of the pDFA, however.
      
      Consider a pDFA over a finite alphabet $A$.
      We may always assume that $A$ is involutive (potentially by passing to $A^{\pm 1}$) and say that the language of a state is \emph{reduced} if it does not contain any word $w$ with $a a^{-1}$ (for any $a \in A$) as a factor.
      A pDFA is \emph{reduced} if the languages of all of its states are, which is equivalent to it not containing a path labeled with $a a^{-1}$ for any $a \in A$.
      \begin{remark*}
        Note that we cannot have a path labeled with $a a^{-1}$ (or $a a^{-1}$) in a reduced pDFA but we can absolutely have a state with two distinct incoming $a$-transitions!
        Thus, being reduced as a pDFA is different to having a deterministic involutive closure.
      \end{remark*}
      \begin{example}\label{ex:pDFAs}
        The automaton in \autoref{fig:ex1} is not deterministic and, thus, not a pDFA; the associated tree is not deterministic.
        The automaton in \autoref{fig:exInvolutive} is deterministic but not reduced (and the associated tree remains not deterministic due to the involutive closure).
        The automaton in \autoref{fig:ex2} finally is a reduced pDFA and $\Gamma(p)$ is a deterministic involutive graph.
      \end{example}
      \begin{theorem}\label{thm:detContextFreeIspDFA}
        Let $\Gamma$ be an involutive tree.
        Then $\Gamma$ is deterministic and context-free if and only if $\Gamma$ is isomorphic (as a rooted, involutive, $A$-graph without node labels) to $\Gamma(p)$ for some state $p$ of a reduced pDFA.
      \end{theorem}
      \begin{proof}
        First, suppose that $\Gamma$ is deterministic and context-free.
        By \autoref{thm:contextFreeIsRegular}, $\Gamma$ is then isomorphic (as a rooted, involutive $A$-graph) to $\Gamma(p)$ for some state $p$ of an mNFRA $\mathcal{A}$ over the finite involutive alphabet $A$.
        Without loss of generality, we may assume that all states of $\mathcal{A}$ are reachable from $p$.
        
        We first show that $\mathcal{A}$ comes from a pDFA using \autoref{fct:pDFAandmNFA}.
        Consider two transitions $\tau$ and $\tau'$ of $\mathcal{A}$ with $\alpha(\tau) = \alpha(\tau') = q$ and $\lambda(\tau) = \lambda(\tau') = a$.
        We need to show $\tau = \tau'$.
        By our assumption, the state $q$ is reachable from $p$ via some run $\varrho$.
        This means that we have a path from $\varepsilon$ to $\varrho$ in $\Gamma(p)$ (labeled by $\lambda(\varrho)$).
        We further have the edges $\aedge{\varrho}{a}{\varrho\tau}$ and $\aedge{\varrho}{a}{\varrho\tau'}$ in $\Gamma(p)$.
        The determinism of $\Gamma(p)$ now yields $\varrho\tau = \varrho\tau'$ and, thus, $\tau = \tau'$.
        
        It remains to show that the pDFA is reduced (for this direction of the proof).
        Suppose to the contrary that we have transitions $\tau_1 = \aedge{q_1}{a}{q_2}$ and $\tau_2 = \aedge{q_2}{a^{-1}}{q_3}$ (for any $a$ from the involutive alphabet $A$).
        Since we may reach $q_1$ from $p$ via a run $\varrho$, we obtain a path from $\varepsilon$ to $\varrho$ in $\Gamma(p)$ together with the edges $\aedge{\varrho}{a}{\varrho\tau_1}$ and $\aedge{\varrho\tau_1}{a^{-1}}{\varrho\tau_1\tau_2}$.
        Since $\Gamma(p)$ is involutive, we must also have the edge $\aedge{\varrho\tau_1}{a^{-1}}{\varrho}$ in $\Gamma(p)$, which yields a contradiction since $\varrho \neq \varrho\tau_1\tau_2$ and $\Gamma(p)$, thus, cannot be deterministic.
        
        For the converse direction consider a state $p$ of a reduced pDFA $\mathcal{A}$ with state set $Q$.
        Since we may consider a pDFA as a special case of an mNFA, we obtain from \autoref{thm:contextFreeIsRegular} that $\Gamma(p)$ is context-free and it remains to show that it is deterministic.
        Suppose we have edges $\laedge{v}{a}{u}$ and $\aedge{u}{a}{v'}$ in $\Gamma(p)$.
        We need to show $v = v'$.
        There are two possible cases: $u$ is the parent of $v$ and $v'$, or $v$ is the parent of $u$ which is the parent of $v'$ (the case with swapped $v$ and $v'$ is symmetric).
        
        In the first case, we have $v = ua = v'$ (by the construction of $\Gamma(p)$ for pDFAs) and are done.
        In the second case, let $q$ be the label of $v$, $r$ be the label of $u$ and $q'$ be the label of $v'$.
        We must have the edge $\laedge{v}{a}{u}$ because of the edge $\aedge{v}{a^{-1}}{v'}$ in the construction of $\Gamma(p)$ (since $v$ is the parent of $u$).
        This only is created by a transition $\aedge{q}{a^{-1}}{r}$ in $\mathcal{A}$.
        Similarly, the edge $\aedge{u}{a}{v}$ is only created by a transition $\aedge{r}{a}{q'}$ in $\mathcal{A}$.
        This implies $a^{-1}a \in \mathscr{L}(q)$, which is not possible since $\mathcal{A}$ is reduced.
      \end{proof}
      \begin{remark}\label{rem:pDFADescription}
        Just like roots of mNFAs are a natural finite description of (general) context-free involutive graphs, we obtain from \autoref{thm:detContextFreeIspDFA} that roots of reduced pDFAs are a finite form of describing deterministic context-free trees (where we require the state to be a root since non-reachable parts of the pDFA do not contribute to the generated tree).
        In some sense, this description is even more natural than the mNFA one for general context-free graphs since it does cover the fact that the tree is deterministic.
        We will, therefore, use this as the description of our choice whenever we need to algorithmically encode deterministic context-free trees.
      \end{remark}
    \end{subsection}
  \end{section}
  
  \begin{section}{The Isomorphism Problem}
    \begin{subsection}{The Rooted Case}\label{ssct:rootedIso}
      We first show that the rooted isomorphism problem for deterministic context-free trees\footnote{Recall that a root of an automaton is a state from which all other states are reachable.}
      \problem{
        a reduced pDFA $\mathcal{A}$ with root $\check{p}$ and\newline
        a reduced pDFA $\mathcal{B}$ with root $\check{q}$%
      }{%
        are $\Gamma(\check{p})$ and $\Gamma(\check{q})$ isomorphic as rooted involutive graphs?
      }\noindent
      is $\NL$-complete.
      In fact, we show that its complement is $\NL$-complete, from which the statement follows by the theorem of Immerman and Szelepcsényi \cite{immerman88, szelepcsenyi87} (also e.\,g.\ \cite[Theorem~7.6]{papadimitriou97computational}).
      By \autoref{fct:isoIsLang}, this is equivalent to showing that checking $\mathscr{L}(\check{p}) \neq \mathscr{L}(\check{q})$ for the same input is $\NL$-complete.
      
      For the membership in $\NL$, we first observe that we may assume both input automata to have a common alphabet $A$ since every pDFA over $A$ is, in particular, also one over $A \cup B$.
      The membership in $\NL$ can then easily be seen by guessing a witness in the symmetric difference of the two languages letter by letter.
      For this, we only need to store the current state $p$ in $\mathcal{A}$ and the current state $q$ in $\mathcal{B}$ and advance them to $p \cdot a$ and $q \cdot a$ for every guessed letter $a \in A$ (if possible).
      If we reach the situation that we can read $a$ only in one of them (but not the other), we have found a word in the symmetric difference and accept.
      
      For $\NL$-hardness, we give a $\LogSpace$-reduction from a version of the graph accessibility problem.
      The main issue here is that we would like to have all states accessible; i.\,e.\ the inputs $\check{p}$ and $\check{q}$ should be roots in their automata (compare to \autoref{rem:pDFADescription}).
      To achieve this, we use an approach similar to \cite[Lemma~2.2]{cho1992parallel}.
      \begin{proposition}
        The rooted isomorphism problem for deterministic context-free trees is $\NL$-hard.
      \end{proposition}
      \begin{proof}
        We reduce the problem \DecProblem{2GAP}
        \problem
          {a finite unlabeled directed graph $\Gamma$ with node set $V = \{ 0, \dots, n - 1\}$ where every node has at most two outgoing edges}
          {is there a directed path from $0$ to $n - 1$ in $\Gamma$?}\noindent
        to the complement of the isomorphism problem.
        One can see that \DecProblem{2GAP} is $\NL$-hard using a master reduction very similar to the classical one to show that the (general) accessibility problem for directed graphs is $\NL$-hard (see e.\,g.\ \cite[Theorem~16.2]{papadimitriou97computational}):
        Instead of starting with an arbitrary $\NL$-machine and then using its configuration graph, we may assume the machine to only have at most two successor configurations.
        This is clearly not a restriction compared to general nondeterminism as a bigger number of nondeterministic branching may simply be achieved by using multiple steps.
        
        For the reduction, consider some $\Gamma$ as given in the input for \DecProblem{2GAP}.
        Without loss of generality, we may assume that $n = 2^\ell$ is a power of two (e.\,g.\ by adding additional states numbered from $1$ to the required amount).
        We may also assume that $0$ has no incoming edges (as dropping those does not affect the existence of a path from $0$ to $n - 1$) and that $n - 1$ does not have any outgoing edges (for the same reason).
        
        We do the reduction in two steps.
        First, we extend $\Gamma$ into a pDFA $\mathcal{A}'$ over the binary alphabet $A = \{ 0, 1 \}$ by adding labels to all edges.
        If a state $i$ has two outgoing edges to the states $j$ and $k$ with $j \leq k$, respectively, we label the edge to $j$ with $0$ and the edge to $k$ with $1$.
        If it has only a single outgoing edge, we duplicate this and label one copy with $0$ and the other copy with $1$.
        If a node $i < n - 1$ does not have any outgoing edges, we add two self-loops $\aedge{i}{0}{i}$ and $\aedge{i}{1}{i}$.
        
        We observe that we have a run starting in $0$ and ending in $n - 1$ in $\mathcal{A}'$ if and only if there was a path from $0$ to $n - 1$ in $\Gamma$.
        Further observe that we can read $0$ and $1$ in all states except for $n - 1$ (where we can read neither).
        Thus, we have $\mathscr{L}(0) = \{ 0, 1 \}^*$ if and only if there is no path from $0$ to $n - 1$ in $\Gamma$.
        
        Note that $\mathcal{A}'$ only uses positive letters from $A$ and is, thus, reduced (as a pDFA over the involutive alphabet $A^{\pm 1}$).
        Also note that $\mathcal{A}'$ can clearly be computed by a deterministic $\LogSpace$-transducer.
        
        We could now use $\mathcal{A}'$ as our first pDFA with state $0$ and construct a second pDFA with a single state and two self-loops for $0$ and $1$ for our reduction, but this has the problem that $n - 1$ is not always accessible from $0$ in $\mathcal{A}'$ and $0$ is, thus, not always a root (which we required for inputs to the isomorphism problem).
        To overcome this, we do the second step of our reduction:
        First, rename every state $i$ in $\mathcal{A}'$ to the binary representation $\operatorname{bin}_\ell i$ of $i$ padded to length $\ell$ with leading zeros.
        For each $w \in \{ 0, 1 \}^{< \ell}$, we add a new state $w$ and the transitions $\aedge{w}{0}{w0}$ and $\aedge{w}{1}{w1}$.
        Denote the thus constructed pDFA by $\mathcal{A}$.
        Observe that every state of $\mathcal{A}$ is accessible from $\varepsilon$ and that $\mathcal{A}$ can be constructed in logarithmic space.
        Further observe that $\mathcal{A}$ remains reduced.
        
        For $\mathcal{B}$ we use a copy of $\mathcal{A}$ where we replace the state $\operatorname{bin}_\ell 0 = 0^\ell$ by a different state $f$ with the transitions $\aedge{f}{0}{f}$ and $\aedge{f}{1}{f}$.
        As the states $p$ in $\mathcal{A}$ and $q$ in $\mathcal{B}$, we use $\varepsilon$ in both cases.
        
        Now, observe that the languages $\mathscr{L}(p)$ and $\mathscr{L}(q)$ can only differ due to the possible runs starting in $0^\ell$ and $f$, respectively (which are reached by reading $0^\ell$ from $p$ and $q$, respectively).
        The language of $f$ is, by construction, $\{ 0, 1 \}^*$ and the language of $0^\ell$ is $\{ 0, 1 \}^*$ if and only if there is no path from $0$ to $n - 1$ in $\Gamma$.
        Thus, we have $\mathscr{L}(p) \neq \mathscr{L}(q)$ if and only if there is no such path, which concludes our proof.
      \end{proof}
      
    \end{subsection}
    \begin{subsection}{The Non-Rooted Case}
      We will next show that the non-rooted isomorphism problem for deterministic context-free trees
      \problem{
        a reduced pDFA $\mathcal{A}$ with state $\check{p}$ and\newline
        a reduced pDFA $\mathcal{B}$ with state $\check{q}$%
      }{%
        are $\Gamma(\check{p})$ and $\Gamma(\check{q})$ isomorphic as non-rooted involutive graphs?
      }\noindent
      is also $\NL$-complete. The result does not change if we additionally assume all states of their respective automtaton to be reachable from $\check{p}$ and $\check{q}$.
      
      \paragraph*{The $\NL$-Algorithm.}
      In order to describe the $\NL$-Algorithm, we fix reduced pDFAs $\mathcal{A} = (P, S)$ and $\mathcal{B} = (Q, T)$ (with $P \cap Q = \emptyset$) over a (without loss of generality; see above) common finite alphabet $A$ and states $\check{p} \in P$ and $\check{q} \in Q$.
      We also introduce the symbol $\check{=}$ to indicate that two $A$-graphs are isomorphic as rooted graphs.
      Furthermore, if $v$ is a node of $\Gamma(s)$ for some state $s$ of a pDFA, we denote by $\Gamma(s, v)$ the (node-labled) rooted involutive tree where we have changed the root to be $v$.
      Also recall that $\Gamma(s) \setminus vA^*$ denotes the graph arising from $\Gamma(s)$ by removing the end-cone of $v$ (which can be considered as the subtree \enquote{below} $v$).
      
      For testing whether $\Gamma(\check{p})$ is isomorphic to $\Gamma(\check{q})$ as a non-rooted graph in $\NL$, we introduce two subroutines $U(p, q)$ and $U(p, q, b_0)$ with the parameters $p \in P$, $q \in Q$ and $b_0 \in A$.
      We will implement these subroutines as $\NL$-algorithms with the semantics given in \autoref{tbl:semantics} (compare also to \autoref{fig:schematicSemantics}).
      \begin{table}
        \begin{tabularx}{\linewidth}{lX}
          Routine & When does it return \texttt{\textbf{true}}? \\\hline
          $U(p, q)$ & There is a node $v$ labeled by $q$ in $\Gamma(\check{q})$ with $\Gamma({p}) \cheq \Gamma(\check{q}, v)$. \\
          $U(p, q, b_0)$ & There is a node $v$ labeled by $q$ in $\Gamma(\check{q})$ with $\Gamma({p}) \cheq \Gamma(\check{q}, v) \setminus vb_0A^*$.
        \end{tabularx}
        \caption{Semantics for the subroutines $U(p, q)$ and $U(p, q, b_0)$}\label{tbl:semantics}
      \end{table}%
      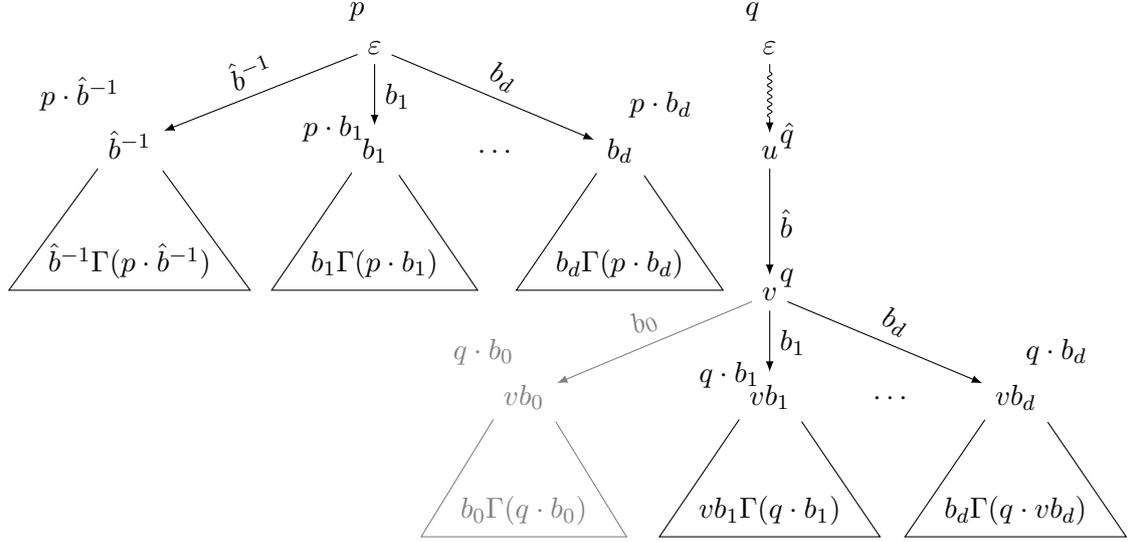
\begin{figure}\centering
        \begin{tikzpicture}[auto, shorten >=1pt, >=latex]
          \node[label={[above left]$p$}] (e) {$\varepsilon$};
          \node[label={[above left]$p\cdot\hat{b}^{-1}$}, below left=1.25cm and 3cm of e, anchor=base] (bi) {$\hat{b}^{-1}$};
          \node[label={[left, anchor=east]$p\cdot b_1$}, below=1.25cm of e, anchor=base] (b1) {$b_1$};
          \node[label={[above right]$p\cdot b_d$}, below right=1.25cm and 3cm of e, anchor=base] (bd) {$b_d$};
          \node[anchor=base] at ($(b1.base)!0.5!(bd.base)$) {$\cdots$};
          
          \path[->] (e) edge node[sloped] {$\hat{b}^{-1}$} (bi)
                        edge node[right] {$b_1$} (b1)
                        edge node[sloped] {$b_d$} (bd);
          
          \node[below=1.5cm of bi.base, anchor=base] (biT) {$\hat{b}^{-1}\Gamma(p\cdot\hat{b}^{-1})$};
          \draw (bi.south west) -- ($(biT.base west)+(-1em,-1.5ex)$) -- ($(biT.base east)+(1em,-1.5ex)$) -- (bi.south east);
          
          \node[below=1.5cm of b1.base, anchor=base] (b1T) {$b_1\Gamma(p\cdot b_1)$};
          \draw (b1.south west) -- ($(b1T.base west)+(-1em,-1.5ex)$) -- ($(b1T.base east)+(1em,-1.5ex)$) -- (b1.south east);
          
          \node[below=1.5cm of bd.base, anchor=base] (bdT) {$b_d\Gamma(p\cdot b_d)$};
          \draw (bd.south west) -- ($(bdT.base west)+(-1em,-1.5ex)$) -- ($(bdT.base east)+(1em,-1.5ex)$) -- (bd.south east);

          \node[label={[above left]$q$}, right=4.75cm of e] (qe) {$\varepsilon$};
          \node[label={[right]$\hat{q}$}, below=1.25cm of qe, anchor=base] (u) {$u$};
          \node[label={[right]${q}$}, below=1.75cm of u, anchor=base] (v) {$v$};
          
          \path[->] (qe) edge[decorate, decoration={snake, pre length=1mm, post length=2mm, amplitude=0.3mm, segment length=1mm}] (u);
          \path[->] (u) edge node {$\hat{b}$} (v);
          
          \begin{scope}[gray]
            \node[label={[above left]$q\cdot b_0$}, below left=1.25cm and 3cm of v, anchor=base] (vb0) {$vb_0$};
            \node[below=1.5cm of vb0.base, anchor=base] (vb0T) {$b_0\Gamma(q\cdot b_0)$};
            \draw (vb0.south west) -- ($(vb0T.base west)+(-1em,-1.5ex)$) -- ($(vb0T.base east)+(1em,-1.5ex)$) -- (vb0.south east);
            \path[->] (v) edge node[sloped] {$b_0$} (vb0);
          \end{scope}
          
          \node[label={[left, anchor=east]$q\cdot b_1$}, below=1.25cm of v, anchor=base] (vb1) {$vb_1$};
          \node[label={[above right]$q\cdot b_d$}, below right=1.25cm and 3cm of v, anchor=base] (vbd) {$vb_d$};
          \node[anchor=base] at ($(vb1.base)!0.5!(vbd.base)$) {$\cdots$};
          
          \path[->] (v) edge node[right] {$b_1$} (vb1)
                        edge node[sloped] {$b_d$} (vbd);

          \node[below=1.5cm of vb1.base, anchor=base] (vb1T) {$vb_1\Gamma(q\cdot b_1)$};
          \draw (vb1.south west) -- ($(vb1T.base west)+(-1em,-1.5ex)$) -- ($(vb1T.base east)+(1em,-1.5ex)$) -- (vb1.south east);
          
          \node[below=1.5cm of vbd.base, anchor=base] (vbdT) {$b_d\Gamma(q\cdot vb_d)$};
          \draw (vbd.south west) -- ($(vbdT.base west)+(-1em,-1.5ex)$) -- ($(vbdT.base east)+(1em,-1.5ex)$) -- (vbd.south east);
        \end{tikzpicture}
        \caption{Schematic drawing for the semantics of $U(p, q)$ and $U(p, q, b_0)$: The left tree must be isomorphic to the right tree by mapping $\varepsilon$ to $v$. For $U(p, q, b_0)$ we ignore the gray subtree for this isomorphism.}
        \label{fig:schematicSemantics}
      \end{figure}%
      In particular, $\Gamma(\check{p})$ and $\Gamma(\check{q})$ are isomorphic as non-rooted trees if $\Gamma(\check{p})$ is isomorphic as a rooted tree to $\Gamma(\check{q}, v)$ for some new root $v$ which has some label $q \in Q$.
      Thus, after proving the semantics from \autoref{tbl:semantics}, we have that $\Gamma(\check{p})$ and $\Gamma(\check{q})$ are isomorphic as non-rooted trees if and only if $U(\check{p}, q)$ returns \texttt{\textbf{true}} for some $q \in Q$.
      As we can guess $q \in Q$ nondeterministically, this means that we have solved the problem when we have implemented $U(p, q)$ as an $\NL$-algorithm with the required semantics.
      
      The implementation may be found in \autoref{alg:U}.
      Both algorithms are very similar and we give them as a single pseudo-code listing.
      Taking only the black part, we obtain the implementation for $U(p, q)$ and we obtain the implementation for $U(p, q, b_0)$ by additionally adding the gray parts.
      We use two special instructions to handle the nondeterminism of the algorithm:
      For a finite set $X$, we use \texttt{\textbf{guess}($X$)} to nondeterministically guess an element $x \in X$ (i.\,e.\ we have a new nondeterministic computational branch for each $x \in X$).
      With \texttt{\textbf{assert}} we fail on all nondeterministic branches for which the condition following the \texttt{\textbf{assert}} is not satisfied.
      This condition can in particular also be another $\NL$-algorithm.
      In this case, the computation survives on those branches for which the inner algorithm terminates and accepts.
      \begin{algorithm}
        \begin{pseudocode}
fun U($p \in P$, $q \in Q$|, $\color{gray}b_0 \in A$|):
  if $q = \check{q}$ and guess($ \{ $true, false$ \} $):(*@\label{ln:if} \hfill \normalfont{$\triangleright$ at the root we do not have a predecessor} @*)
    assert $\out p = \out q \textcolor{gray}{{}\setminus \{ b_0 \}}$(*@ \hfill \normalfont{$\triangleright$ make sure the outgoing transitions at $p$ match} @*)
    for $a \in \out p$:
      assert $\Gamma(p \cdot a) \cheq \Gamma(q \cdot a)$(*@ \hfill \normalfont{$\triangleright$ make sure all subtrees below the children are isomorphic} @*)
    return true
  else:
    $(\hat{q}, \hat{b}) \gets{}$guess($\{ (\hat{q}, \hat{b}) \mid \aedge{\hat{q}}{\hat{b}}{q} \text{ in } T \}$)(*@ \hfill \normalfont{$\triangleright$ guess a predecessor of $q$} @*)
    assert $\out p = \{ \hat{b}^{-1} \} \cup \textcolor{gray}{(} \out q \textcolor{gray}{{}\setminus \{ b_0 \} )}$(*@ \hfill \normalfont{$\triangleright$ make sure the outgoing transitions at $p$ match} @*)
    for $b \in \textcolor{gray}{(} \out q \textcolor{gray}{{}\setminus \{ b_0 \} )}$:
      assert $\Gamma(p \cdot b) \cheq \Gamma(q \cdot b)$(*@ \hfill \normalfont{$\triangleright$ make sure all subtrees below the children are isomorphic} @*)
    return U($p \cdot \hat{b}^{-1}$, $\hat{q}$, $\hat{b}$)(*@\label{ln:return} \hfill \normalfont{$\triangleright$ this is a tail recursion for the subtree below $\hat{b}^{-1}$ } @*)
        \end{pseudocode}
        \caption{Implementation of $U(p, q)$ and $U(p, q, b_0)$ as $\NL$-algorithms}\label{alg:U}
      \end{algorithm}
      
      Before diving deeper into the algorithm, we point out that, at any moment in time, we only need to store constantly many states from $P$ and $Q$ and letters from $A$.
      In particular, the recursion in \autoref{ln:return} is a tail recursion and we do not need to store any previous variable values at this point.
      Since we already know that checking $\Gamma(p) \cheq \Gamma(q)$ is possible in $\NL$ by \autoref{ssct:rootedIso}, this shows that the algorithm only requires logarithmic space.
      
      The algorithm is best understood via its correctness proof.
      For this, we define the \emph{recursion depth} of a computational branch as the number of times we have executed \autoref{ln:return} on it.
      We show a parameterized version of the semantics from \autoref{tbl:semantics}:
      \begin{itemize}[noitemsep]
        \item There is a node $v$ with label $q$ on level $\ell$ in $\Gamma(\check{q})$ such that $\Gamma({p}) \cheq \Gamma(\check{q}, v)$ if and only if $U(p, q)$ accepts on a computational branch of recursion depth $\ell$.
        \item There is a node $v$ with label $q$ on level $\ell$ in $\Gamma(\check{q})$ such that $\Gamma({p}) \cheq \Gamma(\check{q}, v) \setminus vb_0A^*$ if and only if $U(p, q, b_0)$ accepts on a computational branch of recursion depth $\ell$.
      \end{itemize}
      
      We first show the \enquote{only if} part of both statements in parallel by induction on $\ell$.
      First, let $\ell = 0$.
      Suppose we have $\Gamma({p}) \cheq \Gamma(\check{q}, v)$ for some $v$ with label $q$ on level $0$ in $\Gamma(\check{q})$.
      This means that we must have $v = \varepsilon$ (i.\,e.\ the root of $\Gamma(\check{q})$) with label $q = \check{q}$ and, thus, that $\Gamma({p}) \cheq \Gamma(\check{q})$.
      Therefore, on some computational branch, we enter the \texttt{\textbf{if}}-part on \autoref{ln:if}.
      Since $\Gamma({p})$ and $\Gamma(\check{q})$ are isomorphic as rooted trees, we have $\out {p} = \out q$ and that, for every $a \in \out p = \out q$, the subtrees $aA^*$ in $\Gamma(p)$ and in $\Gamma(\check{q})$ are isomorphic (as rooted trees).
      In other words, the two \texttt{\textbf{assert}} statements hold and we finally return \texttt{\textbf{true}} (without any recursion on \autoref{ln:return}).
      The proof for $U(p, q, b_0)$ is the same -- except that we need to ignore the subtree $b_0A^*$ in $\Gamma(\check{q})$.
      
      For the inductive step from $\ell - 1$ to $\ell$, assume $\Gamma(p)\cheq \Gamma(\check{q}, v)$ for the node $v$ on level $\ell$ of $\Gamma(\check{q})$ with label $q$.
      Let $u$ be the parent of $v$ in $\Gamma(\check{q})$ and let it be connected to $v$ by the edge $\aedge{u}{\hat{b}}{v}$; this implies $v = u \hat{b}$.
      It is on level $\ell - 1$ and labeled by some $\hat{q}$ such that there is a transition $\aedge{\hat{q}}{\hat{b}}{q}$ in $\mathcal{B}$ (i.\,e.\ in $T$).
      On some computational branch, we enter the \texttt{\textbf{else}}-branch for the \texttt{\textbf{if}} on \autoref{ln:if} and guess precisely this combination $\hat{q}, \hat{b}$.
      We may draw $\Gamma(p)$ and $\Gamma(\check{q}, v)$ (including the node labels) as in \autoref{fig:GammaPGammaCheckQV}
      \begin{figure}\centering
        \begin{tikzpicture}[auto, shorten >=1pt, >=latex]
          \node[label={[above left]$p$}] (e) {$\varepsilon$};
          \node[label={[above left]$p\cdot\hat{b}^{-1}$}, below left=1.25cm and 1.5cm of e, anchor=base] (bi) {$\hat{b}^{-1}$};
          \node[label={[above right]$p\cdot b$}, below right=1.25cm and 1.5cm of e, anchor=base] (b) {$b$};
          
          \path[->] (e) edge node[sloped] {$\hat{b}^{-1}$} (bi)
                        edge node[sloped] {$b$} (b);
          
          \node[below=1.5cm of bi.base, anchor=base] (biT) {$\hat{b}^{-1}\Gamma(p\cdot\hat{b}^{-1})$};
          \draw (bi.south west) -- ($(biT.base west)+(-1em,-1.5ex)$) -- ($(biT.base east)+(1em,-1.5ex)$) -- (bi.south east);
          
          \node[below=1.5cm of b.base, anchor=base] (bT) {$b\Gamma(p\cdot b)$};
          \draw (b.south west) -- ($(bT.base west)+(-1em,-1.5ex)$) -- ($(bT.base east)+(1em,-1.5ex)$) -- (b.south east);

          \node[label={[above right]$q$}, base right=6cm of e] (v) {$v = u\hat{b}$};
          \node[label={[above left]$\hat{q}$}, below left=1.25cm and 1.25cm of v, anchor=base] (u) {$u$};
          \node[label={[above right]$q\cdot b$}, below right=1.25cm and 1.25cm of v, anchor=base] (vb) {$vb$};
          
          \path[->] (u) edge node[sloped] {$\hat{b}$} (v)
                    (v) edge node[sloped] {$b$} (vb);

          \node[below=1.5cm of u.base, anchor=base] (uT) {$\Gamma(\check{q}, u) \setminus u\hat{b}A^*$};
          \draw (u.south west) -- ($(uT.base west)+(-1em,-1.5ex)$) -- ($(uT.base east)+(1em,-1.5ex)$) -- (u.south east);
          
          \node[below=1.5cm of vb.base, anchor=base] (vbT) {$vb\Gamma(q \cdot b)$};
          \draw (vb.south west) -- ($(vbT.base west)+(-1em,-1.5ex)$) -- ($(vbT.base east)+(1em,-1.5ex)$) -- (vb.south east);
        \end{tikzpicture}
        \caption{Drawing of $\Gamma(p)$ and $\Gamma(\check{q}, v)$. We have the edges and subtrees for all $b \in \out q$}
        \label{fig:GammaPGammaCheckQV}
      \end{figure}
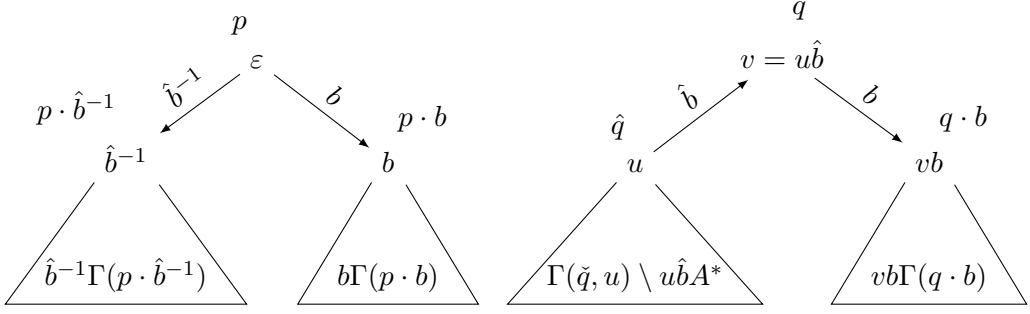
      For the drawing of $\Gamma(\check{q}, v)$, observe that the part to the top and the right of $u$ is precisely the substree $vA^* = u\hat{b}A^*$.
      The rest of the tree exists as a subtree below $u$.
      For the drawing of $\Gamma(p)$, observe:
      We have $\Gamma(p) \cheq \Gamma(\check{q}, v)$ and, thus, $\hat{b}^{-1} \in \out p$ (since $v$ has an outgoing $\hat{b}^{-1}$-edge to $u$ as an involutive graph) and even $\out p = \{ \hat{b}^{-1} \} \cup \out q$ (since they must have the same outgoing edges at the root as isomorphic graphs).
      This is also what we verify in the algorithm in the next step.
      As the trees are isomorphic, also their respective subtrees must be and we have, in particular, $b\Gamma(p \cdot b) \cheq vb\Gamma(q \cdot b)$ for all $b \in \out q$; this is also asserted in the algorithm.
      Finally, we also have $\Gamma(p \cdot \hat{b}^{-1}) \cheq \hat{b}^{-1}\Gamma(p \cdot \hat{b}^{-1}) \cheq \Gamma(\check{q}, u) \setminus u\hat{b}A^*$ for the same reason.
      As $u$ is on level $\ell - 1$ in $\Gamma(\check{q})$, we, thus, have by induction that $U(p \cdot \hat{b}^{-1}, \hat{q}, \hat{b})$ accepts on some computational branch of recursion depth $\ell - 1$, which yields the overall computational branch of recursion depth $\ell$ for $U(p, q)$.
      Again, the proof of $U(p, q, b_0)$ is the same except that we need to ignore the subtree $vb_0A^*$ below $v$ in all arguments and, thus, consider the set $\out q \setminus \{ b_0 \}$ instead of all of $\out q$ (which is precisely what the algorithm does).
      
      The other direction works in the same way:
      Suppose we have an accepting computational branch of recursion depth $0$ for $U(p, q)$.
      This can only happen if we enter the \texttt{\textbf{if}} at \autoref{ln:if}.
      This means that we must have $q = \check{q}$ and we claim that $\Gamma(p) = \Gamma(\check{q})$ (i.\,e.\ the sought node $v$ is the root $\varepsilon$).
      Indeed, we have that the outgoing edges at the root in $\Gamma(p)$ are the same as those at the root of $\Gamma(\check{q})$ since termination of the computation implies $\out p = \out q$.
      The subtrees reached by these edges, respectively, are also isomorphic (as rooted trees) since we must have $\Gamma(p \cdot a) \cheq \Gamma(\check{q} \cdot a)$ for all $a \in \out p$ by termination of the algorithm.
      This shows the claimed isomorphism.
      
      For the inductive step, consider an accepting computational branch for $U(p, q)$ of recursion depth $\ell > 0$.
      Here, we must enter the \texttt{\textbf{else}} branch of the \texttt{\textbf{if}} at \autoref{ln:if}.
      Let $\hat{q}$ be the chosen predecessor of $q$ in $\mathcal{B}$ (with the transition $\aedge{\hat{q}}{\hat{b}}{q} \in T$) on our computational branch.
      Since we have $\out p = \{ \hat{b}^{-1} \} \cup \out q$, we may draw $\Gamma(p)$ as in \autoref{fig:GammaPGammaCheckQV}.
      By induction (using \autoref{ln:return}), there is some $u$ on level $\ell - 1$ with label $\hat{q}$ in $\Gamma(\check{q})$ with $\hat{b}^{-1}\Gamma(p \cdot \hat{b}^{-1}) \cheq \Gamma(p \cdot \hat{b}^{-1}) \cheq \Gamma(\check{q}, u) \setminus u\hat{b}A^*$.
      We have $\hat{b} \in \out \hat{q}$ and, thus, a node $v = u\hat{b}$ with label $q$ as a child of $u$ in $\Gamma(\check{q})$.
      Note that $v$ is on level $\ell$ in $\Gamma(\check{q})$.
      The children of $v$ are (by construction) the nodes $vb$ with the subtrees $vb\Gamma(q \cdot b)$ for all $b \in \out q$.
      In other words, we can also draw $\Gamma(\check{q}, v)$ as in \autoref{fig:GammaPGammaCheckQV}.
      From induction, we already obtained $\hat{b}^{-1}\Gamma(p \cdot \hat{b}^{-1}) \cheq \Gamma(\check{q}, u) \setminus u\hat{b}A^*$.
      We additionally have $b\Gamma(p \cdot b) \cheq \Gamma(p \cdot b) \cheq \Gamma(q \cdot b) = vb\Gamma(q \cdot b)$ for all $b \in \out q$ from the corresponding \texttt{\textbf{assert}} in the algorithm.
      This shows that $\Gamma(p)$ and $\Gamma(\check{q}, v)$ are isomorphic as rooted trees (as required).
      From a terminating computational branch of $U(p, q, b_0)$, we similarly obtain all the ingredients to show $\Gamma(p) \cheq \Gamma(\check{q}, v) \setminus vb_0A^*$, which concludes our proof.
      
      \paragraph*{$\NL$-Hardness.}
      We obtain $\NL$-hardness from a simple reduction from the rooted version of the isomorphism problem from \autoref{ssct:rootedIso}.
      For the reduction, take the pDFAs $\mathcal{A}$ with state $\check{p}$ and $\mathcal{B}$ with state $\check{q}$ and first add another new letter $\top$ to their (without loss of generality) common alphabet $A$ (where the inverse of $\top$ is distinct to $\top$ if we consider it as an involutive alphabet).
      Then, we add new states $\check{p}'$ and $\check{q}'$ to $\mathcal{A}$ and $\mathcal{B}$, respectively, with the transitions $\aedge{\check{p}'}{\top}{\check{p}}$ and $\aedge{\check{p}'}{\top}{\check{p}}$.
      We use the thus modified automata with the states $\check{p}'$ and $\check{q}'$ as the input for the non-rooted isomorphism problem (and observe that all states of their respective automata are reachable from them).
      This is clearly computable in $\LogSpace$.
      Note that the roots of $\Gamma(\check{p}')$ and $\Gamma(\check{q}')$ are the respectively only node in their graph with an outgoing $\top$-edge.
      Thus, any isomorphism  (of non-rooted graphs) must map them to each other and must, thus, be an isomorphism of rooted graphs, which concludes the proof.
    \end{subsection}
  \end{section}
  
  \paragraph*{Acknowledgements.}
  The author would like to thank Mark Kambites for various discussions on topics related of this paper.
  
  \bibliographystyle{plainurl}
  \bibliography{references}
\end{document}